\documentclass{article}

\usepackage{hyperref}
\usepackage{amsmath,amsfonts,amssymb,amsthm}
\usepackage{mathtools}
\usepackage{mathrsfs}
\usepackage{enumerate}
\usepackage[inline]{enumitem}
\usepackage{tikz}
\usepackage{bm}
\usetikzlibrary{graphs,graphs.standard, positioning}
\usepackage{graphicx}
\usepackage{cleveref}
\usepackage{xcolor}
\usepackage{bbm}
\usepackage{booktabs}
\usepackage[ruled]{algorithm2e}

\usepackage{authblk}

\theoremstyle{plain}
\newtheorem{theorem}{Theorem}[section]
\newtheorem{corollary}[theorem]{Corollary}

\theoremstyle{definition}
\newtheorem{definition}[theorem]{Definition}
\newtheorem{remark}[theorem]{Remark}
\newtheorem{example}[theorem]{Example}

\newcommand{\bx}{\mathbf x}
\newcommand{\by}{\mathbf y}

\newcommand{\be}{\mathbf e}

\DeclareMathOperator{\Geo}{Geo}

\title{On the use of Markov chains for epidemic modeling on networks}

\author[1]{Sooyeong Kim\thanks{Contact: kswim2502@gmail.com}}
\author[2]{Jane Breen}
\author[1]{Ekaterina Dudkina}
\author[3]{Federico Poloni}
\author[1]{Emanuele Crisostomi}
\affil[1]{Department of Energy, Systems, Territory and Constructions Engineering, University of Pisa, Pisa, Italy}
\affil[2]{Faculty of Science, Ontario Tech University, Oshawa, Ontario, Canada}
\affil[3]{Department of Computer Science, University of Pisa, Pisa, Italy}

\begin{document}

\maketitle 
\begin{abstract}
We discuss various models for epidemics on networks that rely on Markov chains. Random walks on graphs are often used to predict epidemic spread and to investigate possible control actions to mitigate them. In this study, we demonstrate that they do not fully reflect the dynamics of epidemics, as they overestimate infection times. Accordingly, we explain how Markov chains may still be used to accurately model the virus spread, and to correctly predict infection times. We also provide an algorithm that efficiently estimates infection times via a sampling strategy. Finally, we present a novel indicator based on infection times, and we compare its node ranking properties with other centrality measures based on random walks.

\smallskip
\noindent \textbf{Keywords.}  Markov chains; random walks; disease spread dynamics on networks; mean first passage times; mean infection times; centrality measure
\end{abstract}

\section{Introduction}

Markov chain models of dynamic processes on graphs have been exploited in a number of successful applications, of which the most notable one is perhaps Google's PageRank algorithm \cite{Langville2006}. Other successful applications include the modeling of power grids \cite{Yan2014}, connections of neurons in neuroscience \cite{Sporns2018}, social networks \cite{Vespignani2012}, road networks \cite{Crisostomi2011}, and epidemic networks \cite{yilmaz2020kemeny}.

These models rely upon the idea that a system may be modeled as a random walk, where the system transitions from state to state in discrete time-steps (possibly itself again) with a given probability that only depends on the current node of the system. In the aforementioned examples, this could be an Internet `surfer' who chooses a link from one web-page to pass to the next web-page \cite{levene2002kemeny}, or a car that at each intersection chooses the next road segment \cite{Crisostomi2011}, or in electric circuits the incoming current at each junction is split in the out-going branches. In such examples, known parameters of graph theory and Markov chains may be used to infer interesting properties of the modeled case study, such as (for instance) mean first passage times, or the second-largest eigenvalue modulus, or Kemeny's constant. For example, in \cite{Crisostomi2011}, Kemeny's constant is used as a measure of how well-connected the road network is, and the change in Kemeny's constant upon removal of a road segment from the network is used as a measure of the importance of that road, thus using a well-understood Markov chain parameter associated with the network to make design and control decisions in the system.

The starting point of this manuscript is the observation that the Markovian modeling of the underlying system may not be strictly appropriate in some applications, and that one must be careful to interpret the properties of the random walk on a given graph carefully, rather than naively equating the behaviour of the random walk with the behaviour of the dynamical system it models. Of particular interest is the consideration of spreading phenomena in a network; for example, the spread of disease in a contact network, the dissemination of rumours in an online social network, and so on. In \cite{salathe2010dynamics}, random walks on community structures are used to understand the dynamic of the spread of disease; and independent random walks on networks are considered for understanding information dissemination or infection time in networks \cite{dimitriou2006infection, giakkoupis2019spread, lam2012information, pettarin2010infectious}. 

Consider the spread of a virus in a population represented by a contact network, assuming that the nodes of the network correspond to individuals, and edges between nodes correspond to close contact between individuals. One should note that the spread of the virus in a population does not occur according to a random walk. Indeed, one currently-infected individual may infect more than one of their contacts simultaneously, and after the disease has passed to another individual, the originally-infected person remains infected, with the possibility to spread the disease further in the next step, in tandem with other infected individuals. This particular type of simultaneity of events is a peculiarity of epidemic applications and does not generally occur in the other applications of Markov chains listed above. For example, in modelling the movement of a vehicle in a road network based on aggregate data will mean that a `random' vehicle moves from road to road in a mutually-exclusive fashion, and the sum of vehicles entering one junction is equal to the sum of vehicles leaving the junction. Similarly, in electrical circuits, the algebraic sum of currents flowing into one node is equal to the sum of currents flowing out of that node. Conversely, the mechanism of disease spread is distinct from these applications, but actually other examples may be found that exhibit similar mechanisms to those of the disease spread (e.g., the propagation of fake news in an internet network is an example that is currently drawing significant research interest).

Accordingly, in this manuscript we will examine carefully the possible pitfalls of modelling spreading phenomena via random walks, and provide examples in which relying on indicators calculated based on random walks can lead to incorrect conclusions. Such indicators can provide useful insight into the underlying topology of the network, but we will illustrate that this should not be equated with equivalent insight into the dynamics of the system. 

In Section \ref{sec:math-prelims}, we give the mathematical formulation of Markov chains and random walks on graphs, along with the associated indicators that are used for quantitative and qualitative insights into systems modelled by these. In Section \ref{sec:toy-example}, we provide several motivating examples in which we compare naive interpretations of random-walk parameters in the context of disease spread with simulated outcomes, in order to highlight some of the ways that this model does not capture the true dynamics. In Section \ref{Sec:2^n homogeneous}, we discuss an alternative Markov chain model for disease spread which \emph{does} effectively capture the disease dynamics, by extending the state space to capture the state of every individual in the network at once by representing these by binary vectors (in which a vector entry contains the information of whether the corresponding individual is infected or not). We also define and explore the \emph{mean infection times} as they relate to a network, to determine the time until the infection spreads from one person in the network to a different person in the network. In Section \ref{sec:ranking}, we propose using mean infection times as centrality measures, and discuss how they compare with other random walk indicators that have been used in the literature for control measures such as vaccination or testing.

\section{Preliminaries}\label{sec:math-prelims}

\subsection{Graph theory}
A graph $G$ is a collection of vertices $V$, indexed $v_1, v_2, \ldots, v_n$, with a set of edges $E$ consisting of pairs of vertices $\{v_i, v_j\}$. If $\{u,v\} \in E$, the vertices $u$ and $v$ are said to be \emph{adjacent}; we also write $u\sim v$, and say that $v$ is a \emph{neighbour} of $u$. The \emph{degree} of a vertex $u$, denoted by $\mathrm{deg}(u)$, is the number of neighbours of $u$. The \emph{adjacency matrix} of a graph is the matrix $A(G)=[a_{i,j}]$ such that
\[a_{i,j} = \left\{\begin{array}{cc} 1, & \text{ if } v_i \sim v_j;\\ 0, & \text{ otherwise.}\end{array}\right.\] 

Throughout this article, we consider only graphs which are simple (no loops or multiple edges between vertices), undirected (there is no orientation associated with an edge), and connected (for any pair of vertices $u$ and $v$, it is possible to reach $u$ from $v$ via a sequence of adjacent vertices). We occasionally use `network' and `graph' interchangeably, and also substitute `vertex' for `node'. The convention in the literature is that a \emph{graph} refers to the abstract mathematical object, while a \emph{network} is rooted in the real world in some way.

\subsection{Markov chains}

Suppose that we model a system as a stochastic process in which, at any given time, the system occupies one of a finite number of states $s_1, s_2, \ldots, s_n$, and transitions between these states in discrete-time steps with some fixed transition probabilities; that is, $t_{i,j}$ denotes the probability of occupying $s_j$ in the next time-step, given that the system is currently in $s_i$. This may be represented as a sequence of random variables $\{X(t) \mid t = 0, 1, 2, \ldots \}$, indexed by time-step and taking values from the state space $\{s_1, s_2, \ldots, s_n\}$. Implicit in the description above is the so-called \emph{Markov property}:
\begin{align*}
&\mathbb{P}[X(k+1) = x_{k+1} \mid X(k) = x_k, \ldots, X(1) = x_1, X(0) = x_0]\\
 = &\mathbb{P}[X(k+1) = x_{k+1} \mid X(k) = x_k].
\end{align*}
Note that $\mathbb{P}[A \mid B]$ denotes the conditional probability that event $A$ occurs, given that event $B$ occurs. The above re-states the assertion that the probability the system occupies a given state in the next time-step (at time $t=k+1$) depends only on the current state of the system (at time $t=k$), and not on the state of the system in any previous step.

The probability transition matrix $T=[t_{i,j}]$ is central in the analysis of the behaviour of the Markov chain. The $(i,j)$ entry of $T^k$ gives the probability that the system occupies the state $s_j$ of the chain after exactly $k$ time-steps, given that the initial state was $s_i$. Given an initial probability distribution vector $\mathbf{u} = \begin{bmatrix} u_1 & u_2 & \cdots & u_n\end{bmatrix}$ in which $u_i$ is the probability of occupying $s_i$ initially, the vector $u^\top T^k$ gives the probability distribution across the state space at time $k$. Under certain conditions on the transition matrix $T$ (primitivity), the Perron-Frobenius theorem indicates that as $k\to \infty$, and independently of the initial distribution $u$, $u^\top T^k$ converges to the unique stationary distribution $\bm{\pi}^\top$ of the Markov chain, which may be calculated as the unique left eigenvector of $T$ corresponding to the eigenvalue 1, normalized so that the entries sum to 1. Since $u^\top T^k$ represents the probability distribution at time $k$, this stationary distribution $\bm{\pi}$ represents the long-term probability distribution across the states; that is, $\pi_i$ represents the probability of the system occupying the state $s_i$ in the long run, or the proportion of time spent in $s_i$.

To quantify the short-term behaviour of the system modelled by a Markov chain, we consider \emph{mean first passage times}. For a Markov chain $\{X(t) \mid t = 0, 1, 2, \ldots\}$, the \emph{first passage time} from $s_i$ to $s_j$ is the random variable $F_{i,j}$ taking on the value of $t$ for which $X(t) = s_j$, given that $X(0) = s_i$ and $X(k)\neq s_j$ for all $k=0, 1,\ldots, t-1$. The mean first passage time, then, is the expected value of $F_{i,j}$, denoted $m_{i,j}$. While this definition is probabilistic in nature, it can be shown that the mean first passage times can be computed using the transition matrix $T$:
\begin{equation}\label{eq:mij} m_{i,j} = \left\{\begin{array}{cc} \be_i^\top (I-T_{(j)})^{-1}\mathbbm{1}, & i < j; \\
\be_{i-1}^\top (I-T_{(j)})^{-1}\mathbbm{1}, & i >j.\end{array}\right.\end{equation}
Here $\be_i$ denotes the $i^{\text{th}}$ standard basis vector, $\mathbbm{1}$ the all-ones vector, and $T_{(j)}$ denotes the principal submatrix of $T$ obtained by deleting the $j^{\text{th}}$ row and column. Here $m_{i,j}$ is calculated as the row sum of $(I-T_{(j)})^{-1}$ corresponding to $s_i$. We note that in this article we take the convention that $m_{i,i} = 0$; alternatively we can define the \emph{mean first return time} to $s_i$, which can be shown to be equal to $\frac{1}{\pi_i}$. The \emph{matrix of mean first passage times} is the matrix $M=[m_{i,j}]$, and while each entry can be computed as above, it is also well-known that 
\[M = (I-Z+JZ_{dg})W^{-1},\]
where $Z$ is the so-called \emph{fundamental matrix} of the Markov chain (see \cite{kemenysnell}), $Z_{dg}$ is the diagonal matrix whose entries consist of the diagonal entries of $Z$, and $W$ is the diagonal matrix whose entries consist of the entries of $\bm{\pi}$.

Given an irreducible Markov chain with transition matrix $T$, stationary vector $\bm{\pi}$, and mean first passage matrix $M = [m_{i,j}]$, one can define, for a fixed index $i$, the quantity
\[\kappa_i = \sum_{j\neq i}^n \pi_jm_{i,j}.\]
This can be interpreted as the expected time to reach a randomly-chosen state $j$, starting from a fixed state $i$. Introduced in the 1960s in \cite{kemenysnell}, this was remarkably shown to be independent of the choice of initial state $i$. As such, it is named Kemeny's constant, and is denoted as $\mathcal{K}(T)$. Noting that $\bm{\pi}^\top \mathbbm{1} = 1$, it can be shown that
\begin{equation}\label{kemeny_formula}\mathcal{K}(T) = \sum_{i=1}^n\sum_{j\neq i}^n \pi_im_{i,j}\pi_j,\end{equation}
admitting the interpretation of $\mathcal{K}(T)$ as the expected time of a random trip in the Markov chain, where the initial and terminal states of the trip are chosen at random, with respect to the stationary distribution.

A random walk on a given graph $G=(V, E)$ is an example of a Markov chain. A random walker traverses the vertices of $G$, at each step choosing an adjacent vertex to move to uniformly at random. Thus the state space consists of the vertices $v_1, \ldots, v_n$, and $t_{i,j} = \frac{1}{\deg(v_i)}$ whenever $\{i,j\}\in E$. Letting $D$ be the diagonal matrix of vertex degrees, and $A(G)$ the adjacency matrix of $G$, the probability transition matrix is given by $T = D^{-1}A(G)$. Note that the stationary distribution vector is 
\[\bm{\pi} = \tfrac{1}{2|E|}\begin{bmatrix} \deg(v_1) & \deg(v_2) & \cdots & \deg(v_n) \end{bmatrix}^\top.\] 
Kemeny's constant for the simple random walk on a graph $G$ can be interpreted as a graph invariant indicating the `connectedness' of the graph \cite{Crisostomi2011, kirkland2016kemeny}, or how fast information `mixes' in the graph \cite{hunter2006mixing, kirkland2010fastest}.

\subsection{Absorbing Markov chains}

A state $s_j$ of a Markov chain is called \emph{absorbing} if $t_{jj}=1$; thus when the chain enters state $s_j$, it remains there in every subsequent time-step. A Markov chain is called an \emph{absorbing Markov chain} if its state space contains at least one absorbing state, and if for every $s_i$ which is not absorbing, there exists some absorbing state $s_j$ and some $k>0$ such that $t_{i,j}^{(k)}>0$; that is, there is a nonzero probability that the Markov chain will reach one of the absorbing states in finite time, for any the initial state. While the analysis of irreducible Markov chains focuses on the limiting stationary distribution and mean first passage times, for absorbing Markov chains the limiting behaviour is always that the chain is eventually absorbed into one of the absorbing states, and one concerns oneself with the probabilities of absorption into the different absorbing states, and the expected time to absorption from some initial transient state.

Suppose the states are indexed so that the absorbing states are listed last in the ordering. This produces a block transition matrix as follows: \[T = \left[\begin{array}{c|c} Q & R \\\hline  O & I\end{array}\right], \] 
where $Q$ is a square matrix representing transitions between transient states, $R$ represents transitions from transient states to absorbing states, and $I$ is the identity matrix, whose order is determined by the number of absorbing states. The analysis of the behaviour of the chain before absorption occurs centres around the computation of the matrix $(I-Q)^{-1}$, the so-called \emph{fundamental matrix for absorbing chains} (see \cite{kemenysnell}). Since 
\[(I-Q)^{-1} = I + Q + Q^2 + \cdots,\]
the $(i,j)$ entry of this matrix captures the expected number of visits to the $j^{\text{th}}$ transient state before absorption, given that the chain starts in the $i^{\text{th}}$ transient state. Thus the $i^{\text{th}}$ row of $(I-Q)^{-1}R$ produces the probability distribution for the eventual absorbing state the chain ends up in, given that it starts in the $i^{\text{th}}$ transient state, and the expected time to absorption given that the chain starts in the $i^{\text{th}}$ transient state is computed as the $i^{\text{th}}$ row sum of $(I-Q)^{-1}$, or $\be_i^\top (I-Q)^{-1}\mathbbm{1}$.

\begin{remark}\label{rem:mfp_abs}
Note that the expression for mean first passage times in \eqref{eq:mij} can be derived using absorbing chain techniques, by designating the state $s_j$ as an absorbing state and replacing the $j^{\text{th}}$ row of $T$ by zeros with a 1 in the $j^{\text{th}}$ position. Then the mean first passage times can be computed as the expected time to absorption, where the matrix $Q$ is the principal submatrix $T_{(j)}$. Furthermore, if one requires the expected time to reach a collection $\mathcal{J}$ of states, this can be determined by the appropriate row sum of $(I-T_{(\mathcal{J})})^{-1}$, where by $T_{(\mathcal{J})}$ we denote the principal submatrix of $T$ with the rows and columns indexed by $\mathcal{J}$ removed.
\end{remark}

\subsection{Epidemic modelling}

In this article our focus is on exploring the ability of Markov chain models to reflect the dynamics of virus spreading in populations. For this purpose we shall only consider elementary SI epidemic models for the disease, where individuals may either be susceptible (S) or infected (I); and once a susceptible is infected, it remains infected indefinitely \cite{daley2001epidemic}. If two individuals are in contact and one is susceptible and the other infected, there is a fixed infection probability $\beta$ that the susceptible individual also becomes infected. Throughout this work, when simulations are done we assume the infection probability is $\beta=0.1$.

Compartmental models of disease spread such as SI models traditionally made the assumption that the population is \emph{well-mixed}, which means that every individual comes into contact with every other individual (making $\beta$ the rate of transmission). In a network representation of the population, each individual is represented by a node and an infected individual may infect another susceptible individual in the next time-step if the two of them are neighbours. This underlying graph is a contact network determining which individuals in the population have been in close contact that is sufficient to facilitate the spread of disease, and can be determined from data in some way (see for example \cite{Genois2018}).

Given a connected graph and at least one initially-infected individual, at some point all of the individuals in the network will be infected as well. The focus of our analysis of the above model of epidemics is on quantities dictating how fast the disease will spread. By considering quantities such as the time until the entire network is infected, or the time until a particular individual becomes infected, we can better measure and understand the role of an individual node in the dynamical process. We note that these concepts may apply naturally to other domains such as fake news spreading in online social networks, as similar epidemic models have been used in \cite{d2021stiffness, shrivastava2020defensive} regarding such applications.

\section{A comparison of random walk metrics and disease simulations} 
\label{sec:toy-example}

In this section, we motivate our work by examining several examples and comparing the dynamics of a random walk on a graph with the dynamics of a disease spreading in the same graph, obtained via simulations. We first consider the graph in Figure \ref{Simple_Graph_1} as a supporting example to clarify our discussion.

\begin{figure}[h!]
\begin{center}
	\begin{tikzpicture}[scale=1]
	\tikzset{enclosed/.style={draw, circle, inner sep=0pt, minimum size=.10cm, fill=black}}
	
	\node[enclosed, label={above, yshift=0cm: $1$}] (v_1) at (-2,0.5) {};
	\node[enclosed, label={below, yshift=0cm: $2$}] (v_2) at (-2,-0.5) {};
	\node[enclosed, label={above, yshift=0cm: $3$}] (v_3) at (-1,0.5) {};
	\node[enclosed, label={below, yshift=0cm: $4$}] (v_4) at (-1,-0.5) {};
	
	\draw (v_1) -- (v_2);
	\draw (v_2) -- (v_3);
	\draw (v_3) -- (v_1);
	\draw (v_1) -- (v_4);
	\end{tikzpicture}
\end{center}
	\caption{Simple graph with 4 nodes.}
	\label{Simple_Graph_1}	
\end{figure}
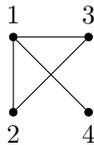
The probability transition matrix for the simple random walk on this graph is
\begin{align*}
T_{RW} = \begin{bmatrix}
0 & 1/3 & 1/3 & 1/3\\
1/2 & 0 & 1/2 & 0\\
1/2 & 1/2 & 0 & 0\\
1 & 0 & 0 &0
\end{bmatrix}.
\end{align*}

In comparing this random walk with the movement of disease, there are some issues to note. If the $(i,j)$ entry of $T_{RW}$ is naively interpreted as the probability of the disease spreading from individual $i$ to individual $j$, then immediately we see that a contact of the same duration between individuals $1$ and $4$ produces two different probabilities $t_{14}$ and $t_{41}$, due to the different degrees of the two nodes.

To overcome the problems of having contact of the same durations associated with different probabilities of spreading the virus, we can consider modifying the transition matrix to create a symmetric one, which we denote $T_{sym}$, adapting the example of \cite{dudkina2021node}. The following definition can be used to account for the infection probability and constructs a symmetric probability transition matrix, so that the probability of $i$ to infect $j$ is the same as that of $j$ to infect $i$. In fact, the matrix is defined so that every off-diagonal entry is equal. This matrix is necessarily doubly-stochastic (both the rows and columns sum to 1). We define
\begin{align*}
	T_{sym}=\alpha A(G)+\left(I - \alpha D\right),
\end{align*}
for some $0<\alpha\leq \frac{1}{d_{\max}}$ where $d_{\max}$ is the maximum node degree, and $D$ is the diagonal matrix of vertex degrees. Note that $\alpha A(G)$ is a substochastic matrix with zero diagonal, while $I-\alpha D$ is a diagonal matrix. One can interpret this definition as a transition matrix for a random walk on a graph where, at each vertex, the random walker has some fixed constant probability $\alpha$ of choosing any of the neighbours of the current vertex, but with some residual probability can choose instead to remain in place. 
For the graph in Figure \ref{Simple_Graph_1}, choosing $\alpha=0.1$, we have
\begin{align*}
T_{sym} = \begin{bmatrix}
0.7 & 0.1 & 0.1 & 0.1\\
0.1 & 0.8 & 0.1 & 0\\
0.1 & 0.1 & 0.8 & 0\\
0.1 & 0 & 0 & 0.9
\end{bmatrix}.
\end{align*}
We note that in this example, every off-diagonal entry is equal to the infection probability $\beta=0.1$ that we wish to use in the disease model. In the case that the desired infection probability is larger than $\frac{1}{d_{\max}}$, it is not possible to choose $t_{ij} = \beta$ for transition probabilities between individuals joined by an edge. However, we remark that mean first passage times between distinct states for transition matrices of this form scale in a natural way with varying choices of $\alpha$. That is, if $T_1 = \alpha_1 A(G) + (I-\alpha_1 D)$ and $T_2 = \alpha_2 A(G) + (I-\alpha_2 D)$, then for $i\neq j$, $m_{i,j}^{(T_1)} = \frac{\alpha_2}{\alpha_1} m_{i,j}^{(T_2)}$.

We now intend to evaluate the accuracy of the models based on random walks. For this purpose, we compute the mean first passage times $m_{i,j}$ for a random walk on the graph. To compare, we also extensively simulate the spreading of the virus, and compute from simulations the average time $M_{i,j}$ for individual $j$ to be infected, given that we start with only individual $i$ infected. 
These are the values which are compared in order to examine the efficacy of random-walk-based methods. 

Let $M_{RW}$ and $M_{sym}$ be the matrices of MFPTs for the random walk and the symmetric random walk on the graph in \Cref{Simple_Graph_1}, respectively. Then, $M_{RW}$ and $M_{sym}$ can be computed in the usual way to produce:
\begin{align*}
M_{RW} = \begin{bmatrix}
  0 & 3.3333 & 3.3333 &7\\
2 & 0 & 2.6667 &9\\
2 & 2.6667 & 0 &9\\
1 & 4.3333 & 4.3333 &0
\end{bmatrix},\;\; M_{sym}=\begin{bmatrix}
  0 & 16.6667 & 16.6667 &30\\
10 & 0 & 13.3333 &40\\
10 & 13.3333 & 0 &40\\
10 & 26.6667 & 26.6667 &0
\end{bmatrix}.
\end{align*}

On the other hand, we can estimate the time for an initially-infected individual $i$ to cause individual $j$ to become infected using simulations. Assuming that the probability of spreading the virus is equal to $0.1$, 
we ran $100000$ Monte Carlo simulations with Algorithm~\ref{algo:MCS}.
\begin{algorithm}
	\KwData{A graph $G = (V, E)$; infection probability $\beta$; a number of iterations $N$.}
	$M_{inf} \gets 0$\; 
	\For{$k = 1,2,\dots,N$}{
		\For{ $i \in V$}{
			$\bx \gets e_i$; $t\gets 0$\;
			\While{$\bx \neq \mathbbm{1}$}{
			$t\gets t+1$\;
			\For {$\{v_l,v_j\}\in E$ with $x_l=1$ and $x_j=0$}{
			Assign $1$ to $x_j$ with probability $\beta$\;
			\If{$x_j =1$}{
				$M^{(k)}_{i,j}\gets t$\;}}
		}}
		$M_{inf} \gets M_{inf} + M^{(k)}$\;
	}
	$M_{inf} \gets M_{inf} / N$;
	\caption{Monte Carlo estimation of $M_{inf}=[M_{i,j}]$} \label{algo:MCS}
\end{algorithm}
We estimated a matrix $M_{inf}$ of these expected times regarding infection as follows:

\begin{align}\label{matrix:M1}
M_{inf}\approx\begin{bmatrix}
 0 & 7.7922 & 7.7766 &10.0051\\
7.7503 & 0 & 7.7358 &17.7442\\
7.7383 & 7.7284 & 0 &17.7707\\
10.0059 & 17.7454 & 17.7628 &0\\
\end{bmatrix}.
\end{align}

Apart from scaling factors (one could normalize $M_{RW}$ or $M_{sym}$ for comparison purposes), it is still obvious that there is a disparity between the estimated values and the theoretical mean first passage times associated with the simple or symmetric random walks. Different predictions are obtained in general; as a single example, the ratio of $m_{1,3}$ to $m_{1,4}$ in $M_{RW}$ and the ratio of the corresponding entries in $M_{sym}$ are both close to $2$, while the ratio of the corresponding entries in $M_{inf}$ is close to $1$.

We consider another example which arises in the context of traffic flow in urban road networks, another domain of application in which random walks have been shown to be an effective model. Let $L(m,n)$ denote an $m\times n$ lattice graph or grid graph; see Figure~\ref{fig:L(3,4)}. We now consider a random walk on this lattice, under the assumption that the next state is chosen among the available neighbours with the same probability $p = 0.1$, while with the residual probability the state will not change (i.e., the random walker stays in place). For instance, this may correspond to a vehicle travelling in a lattice-like road network, where edges correspond to roads and nodes to intersections. At each intersection, the car will choose the next road with the same probability 0.1, and with the residual probability will remain in the same place. For example, for nodes in the middle of the lattice $L(3, 4)$ in Figure~\ref{fig:L(3,4)} which have 4 neighbours, this corresponds to saying that with probability 0.6 the state will not change in one step, and with probability 0.4 it will change state, with the same probability 0.1 of choosing any one of the four neighbours. In this case, it was shown in \cite{Crisostomi2011} that the mean first passage times (MFPTs) to travel from any state to any other state according to such a random walk, can be computed by using the transition matrix $T$ given by $T=0.1A(G)+(I-0.1D)$ where $G$ is $L(3,4)$. The values of such MFPTs are depicted on the left of Figure \ref{Example_MFPTs}.

\begin{figure}[h!]
	\begin{center}
		\includegraphics[width=0.5\textwidth]{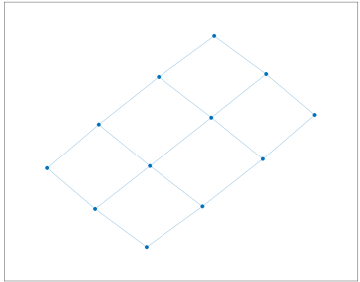}
		\caption{The graph $L(3,4)$}\label{fig:L(3,4)}
	\end{center}
\end{figure}

We now consider the same graph, but we assume now that the graph depicts the social interactions between individuals, and we compute how long it would take for a virus to spread from individual $i$ to individual $j$ in reality. We assume again that the probability of spreading the virus in one step is $\beta = 0.1$, but in this context, one individual can infect more than one neighbour in a single step. Accordingly, the mean time to infection from one individual to any other individual are shown on the right of Figure \ref{Example_MFPTs}. Mean infection times (MITs) from $i$ to $j$ are computed using Algorithm \ref{algo:MCS} under the assumption the $i$ is the only infected individual at the first time-step, and it is straightforward to see that such times are much smaller than those that would have been estimated through a conventional random walk model.

\begin{figure}[h!]
	\begin{center}
		\includegraphics[width=1\textwidth]{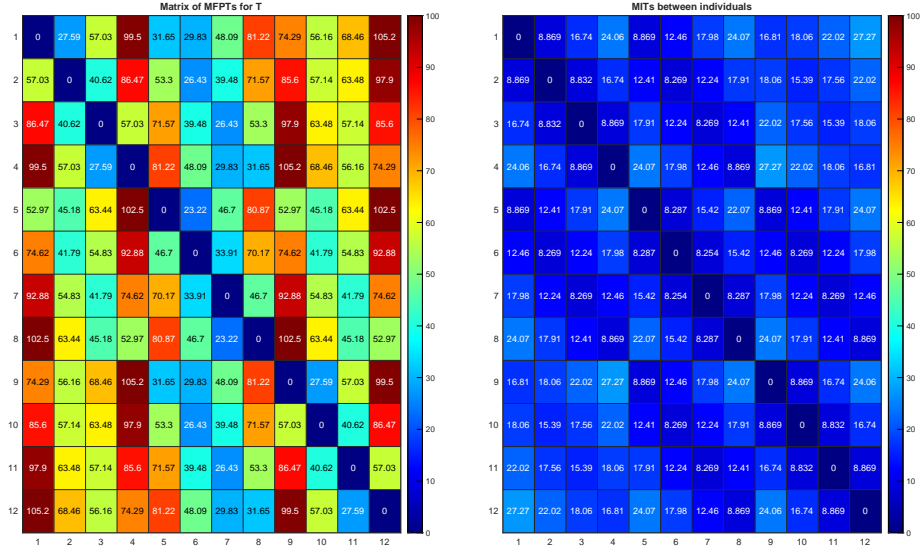}
		\caption{Comparison of mean first passage times for $T$ and mean infection times in the graph $L(3,4)$.}\label{Example_MFPTs}
	\end{center}
\end{figure}

Another example that displays clearly how MFPTs and MITs are qualitatively different is the following. Consider a complete graph on $k$ nodes, and to one of these nodes (which we shall call $A$) attach a new node (which we shall call $B$). An example with $k=5$ is depicted in Figure~\ref{fig:clique}.
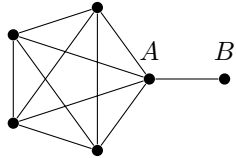
\begin{figure}[h!]
	\centering
	\begin{tikzpicture}
	\node[circle,fill,inner sep=1.5pt] (A) at (0:1) {};
	\node[circle,fill,inner sep=1.5pt] (C) at (72:1) {};
	\node[circle,fill,inner sep=1.5pt] (D) at (144:1) {};
	\node[circle,fill,inner sep=1.5pt] (E) at (216:1) {};
	\node[circle,fill,inner sep=1.5pt] (F) at (288:1) {};
	
	\draw (A) -- (C) -- (D) -- (E) -- (F) -- (A);
	\draw (A) -- (D) -- (F) -- (C) -- (E) -- (A);
	
	\node[circle,fill,inner sep=1.5pt] (B) at (0:2) {};
	
	\draw (A) -- (B);			
	
	\node[above = 1.5 pt of A] {$A$};
	\node[above = 1.5 pt of B] {$B$};
	
	\end{tikzpicture}
	\caption{A graph composed of a clique of $k=5$ nodes attached to another node.}\label{fig:clique}
\end{figure}

Suppose that $A$ is the only initially infected individual. Then, at each time-step $B$ has a probability $\beta$ of becoming infected, hence the MIT from $A$ to $B$ is the expected value of a geometric distribution with parameter $\beta$, that is, $\mathbb{E}[\Geo(\beta)] = \frac{1}{\beta}$. This probability is independent of $k$. On the other hand, the mean time needed for a random walk starting from $A$ to reach $B$ increases as $k$ grows: indeed, for large $k$, the random walk on the clique takes more time to return to $A$, since there are more nodes available to visit. More precisely, one can compute the MFPT from $A$ to $B$ to be $k^2-k+1$.\footnote{
	Let $M$ be the mean first passage time from $A$ to $B$. Starting from $A$, the random walk either reaches $B$ immediately in time $1$, with probability $\frac{1}{k}$, or moves to another node of the clique with probability $\frac{k-1}{k}$. Then, at each subsequent time-step, it has probability $\frac{1}{k-1}$ to return to $A$, hence the time of first return to $A$ from another node of the clique is $\mathbb{E}[\Geo(\frac{1}{k-1})] = k-1$. After the process has returned to $A$, by the Markov property we need again an average time $M$ to reach $B$. This argument produces the equation $M = \frac{1}{k} 1 + \frac{k-1}{k}(1 + (k-1) + M)$, which can be solved for $M$.
}

\subsection{Discussion}

Even in a very simple example, it is obvious to see that modelling the spreading of the virus in a population as a random walk leads to some coarse approximations in terms of the expected mean first passage times. This is rather important to note, as many researchers have used indicators based on random walks, and MFPTs explicitly, in a number of epidemic applications (e.g., vaccination or testing strategies). For instance, this comment holds for indicators such as Kemeny's constant \cite{yilmaz2020kemeny}, or the random walk centrality (RWC) \cite{noh2004random}. Also, random walk betweenness (RWB), originally introduced by Newman \cite{newman2005measure} is a classic indicator based on random walks that is known to be one of the best indicators for which individuals should be vaccinated in a population \cite{salathe2010dynamics}.

All such indicators may still provide interesting insights in the dynamics of a virus based on the underlying structure of the contact network---for example, in \cite{yilmaz2020kemeny} it is shown that Kemeny's constant can be used to detect bridges between disparate communities in contact networks. The effectiveness of these indicators lies in the fact that they are obviously based on the topological structure of a population and of its average contacts; still, one may wonder about the impact of decisions that could be made on more precise models of the disease dynamics, rather than simply the network structure which is revealed through random walk methods. While we do not come to a conclusive answer to that question in this article, we provide the tools by which one can appropriately model and measure the disease spread using a different Markov chain model, and we explore some examples for which random walk indicators disagree with those determined from the more intuitive model.

\section{Alternative Markov chain models and mean infection times}
\label{Sec:2^n homogeneous}

We suggest that the primary reason that a random walk seems to be ineffective as a model for disease spread in a network is that previously-infected individuals remain infectious throughout the process, continuing to infect their contacts and affecting the overall dynamics of the disease in future time-steps. As such, we require a model which retains information about every individual's status at once.  In this section, we describe such a model which is known in the literature, and which has been used in a variety of settings and in many forms. We will review some of these results, and our contribution to the literature is to then define and investigate an analogue of mean first passage times in this setting. We determine how to compute these exactly (the values which were simulated in Section \ref{sec:toy-example}), and discuss some computational issues. We then proceed in Section~\ref{sec:ranking} to establish how these quantities can be developed to determine a centrality indicator for use in control measures in disease spread settings.

\subsection{$(0,1)$ model of epidemic dynamics}
Let $G=(V, E)$ be a graph, and suppose its vertices are labelled $v_1, v_2, \ldots, v_n$. Let $\mathbf{X}(t)=(X_1(t),\dots,X_n(t))\in\{0,1\}^n$ where $X_i(t)=1$ if node $v_i$ is infected at time $t$, and $X_i(t)=0$, otherwise. Since we consider the SI model, we assume that the probability of recovery is $0$; hence if $X_i(t_0) = 1$, then $X_i(t) = 1$ for all $t> t_0$ (models with nonzero recovery probabilities are used in \cite{ahn2014mixing, ganesh2005effect, van2008virus}). Assume that in any contact between a susceptible individual and an infected individual, the probability of infection is a constant $\beta$. Then, $\{\mathbf{X}(t)\mid t=0,1,\dots\}$ is a discrete-time, time-homogeneous Markov chain with a finite state space $S=\{0,1\}^n$; that is, each state is represented by a binary vector in which the $k$th entry is 1 if $v_k$ is infected, and 0 if $v_k$ is suspectible. 

For a subset $A\subset V$, we denote by $s_A$ the vector in $S$ whose $k^\text{th}$ component is $1$ if $v_k\in A$, and zero if $v_k\notin A$. For simplicity, we write $s_{\{v_i\}}$ as $s_i$. In what follows, we assume that $\mathbb{P}[\mathbf{X}(t+1)=s_i\mid\mathbf{X}(t)=(0,\dots,0)]=\frac{1}{n}$ for $1\leq i\leq n$, so that if initially no one is infected, an individual is chosen uniformly at random to be infected in the next step. Then, $\{\mathbf{X}(t)\}$ is an absorbing Markov chain with exactly one absorbing state $s_{V}=(1,\dots,1)$.

Observe that if node $v_i$ has $n_i$ infected neighbours at time $t$, the probability it is not infected at time $t+1$ is $(1-\beta)^{n_i}$. Then, for $\bx,\by\in S$, the transition probability from $\bx$ to $\by$ is given by
\begin{align} \label{transprob}
\mathbb{P}[\mathbf{X}(t+1)=\by \mid \mathbf{X}(t)=\bx]=\prod_{i=1}^{n}\mathbb{P}[X_i(t+1)=y_i\mid \mathbf{X}(t)=\bx]
\end{align}
where
\begin{align} \label{transprob2}
\mathbb{P}[X_i(t+1)=y_i\mid \mathbf{X}(t)=\bx]=&\begin{cases*}
1, & \text{if $y_i=x_i=1$;}\\
1-(1-\beta)^{n_i}, & \text{if $y_i=1, x_i=0$;}\\
(1-\beta)^{n_i}, & \text{if $y_i=0, x_i=0$;}\\
0, & \text{if $y_i=0, x_i=1$.}
\end{cases*}
\end{align}
We note that if $\by$ does not have at least the same infected population as $\bx$, then $P[\mathbf{X}(t+1)=\by\mid \mathbf{X}(t)=\bx]=0$. This means that the states can be ordered in such a way that the probability transition matrix $T$ is an upper triangular matrix.

We now provide a method to obtain the expected time for a group of some infected individuals to infect---directly or indirectly---another group of susceptible individuals; that is, given a subset of individuals $A$ who are initially infected, we compute the expected time until a subset $B$ of individuals are infected, where $B \supset A$. If we choose $A= \{v_i\}$ and $B = \{v_i, v_j\}$, this enables us to exactly compute the matrix $M_{inf}$ in \eqref{matrix:M1}. We define terminology for such expected times in order to distinguish from the mean first passage times used in the previous section.
 
\begin{definition}
	Let $G$ be a connected graph, and let $A,B\subseteq V$ with $A\subset B$. The \textit{mean infection time (MIT)} from $A$ to $B$, denoted as $\mu_{A, B}$, is the expected time for all nodes in $B\setminus A$ to get infected, given that all nodes in $A$ are initially infected. As an analogous concept to the mean first passage matrix, we define the \textit{MIT matrix} to be the matrix $M_{inf}=[M_{i,j}]$ where $M_{i,j}=\mu_{\{v_i\},\{v_i,v_j\}}$. 
\end{definition}

We shall consider the mean infection time from $A$ to $B$ where $A,B\subseteq V$ with $A\subset B$. Let us consider the subset $S_B$ of $S$ that consists of binary vectors with ones in the positions corresponding to $B$; that is, states in which at least the vertices corresponding to $B$ are infected, and possibly others infected too. Then, the mean infection time from $A$ to $B$ is the expected value of the first time (hitting time) at which the system arrives at one of states in $S_B$, starting from the state $s_A$. As discussed in Remark~\ref{rem:mfp_abs}, we have
\begin{align}\label{MIT1}
	\mu_{A, B}=\be_{s_A}^\top(I-T_{(S_B)})^{-1}\mathbbm{1}
\end{align}
where $\be_{s_A}$ is the column vector whose component corresponding to the state $s_A$ is $1$ and zeros elsewhere; $T_{(S_B)}$ is the principal submatrix of $T$ obtained by removing rows and columns corresponding to states in $S_B$; and $\mathbbm{1}$ is the all-ones vector. 

\begin{remark}
The Markov chain above with $2^n$ binary states has been studied before in several contexts; it has been used before to model disease spread, and additional assumptions such as nonzero recovery probability and continuous time have been studied. In most disease spread contexts, the goal is to understand global features such as the conditions under which the process terminates \cite{ahn2014mixing,pare2018analysis,pettarin2010infectious,van2008virus}. To the best of our knowledge, these models have not been used to examine the role of an individual node in a network for disease spread. That is, the expected time to absorption from the state in which a single individual in the network is infected (which indicates individuals' capabilities to spread virus) has not been studied in detail before now using this type of model. Determining the influence of an individual in a contact network is well-studied, particularly in the context of control via testing or vaccination strategies (see \cite{salathe2010dynamics, yilmaz2020kemeny}). Furthermore, the mean infection times are a newly-introduced concept.
\end{remark}

Due to the increased size of the state space (from $n$ states to $2^n$ states), the calculation of these mean infection times becomes computationally expensive. In the remainder of this section, we discuss some computing strategies.

We first note that we need not consider the entire matrix $T_{(S_B)}$ for finding $\mu_{A, B}$. In an SI epidemic model, for any state $\bx$ such that $x_k=0$ for some $k\in A$, the transition probability from $s_A$ to $\bx$ is $0$. That is, if $S_A$ is the set of states consisting of binary vectors with ones in the positions corresponding to $A$, and $\overline{S}_A = S\setminus S_A$ is its set complement, then if the initial state is $s_A$ we may remove from consideration any states of the Markov chain from $\overline{S}_A$ without affecting the computation. These correspond to states in which some individual who was initially infected is now not infected, which is impossible in the SI model. Hence, instead of $T_{(S_B)}$ in \eqref{MIT1}, we may use the matrix $T_{(\overline{S}_A\cup S_B)}$ obtained from $T$ by deleting rows and columns corresponding to the states in $\overline{S}_A\cup S_B$. This substochastic matrix has rows and columns indexed by the states in $S_A$.
Therefore,
\begin{align}\label{MIT}
\mu_{A, B} =\be_{s_A}^T(I-T_{(\overline{S}_A\cup S_B)})^{-1}\mathbbm{1}.
\end{align}
Note that $\be_{s_A}$, $I$ and $\mathbbm{1}$ are re-sized appropriately (their order determined in context), and that $\be_{s_A}$ is the standard basis vector with a 1 in the position corresponding to the state $s_A$, according to where it appears in the reduced state list.

Now let us consider a mean infection time $\mu_{\{v_i\}, V}$. The size of the matrix $T$ is exponential, so it is impractical to store data for $T$ when the number of nodes in the corresponding graph is sufficiently large. Note that $T$ is an upper triangular matrix, and computational experiments show that the sparsity of $T$ is between $1\%$ and $5\%$ for randomly-chosen connected graphs of order $10$. Thus one could consider using a Sparse Triangular Matrix Solver (SpTrSV) (see \cite{ding2020leveraging} for a brief introduction and algorithms). In order to implement SpTrSV, it requires particular data formats \cite{bulucc2009parallel} such as the compressed sparse row (CSR) or the compressed sparse column (CSC), which are two formats for storing nonzero entries of a sparse matrix into three row vectors. However, as the size of the matrix increases exponentially, the number of nonzero entries would increase exponentially since the graphs we consider are connected. Hence, SpTrSV would eventually become infeasible for the computation of mean infection times as in \eqref{Formula:MFPT}.
	
In order to see how using the transition matrix $T$ might not be practical, we consider a particular example of a tree, which is a minimally-connected graph (that is, it has the least number of edges necessary to be a connected graph). Let $G$ be a star with vertex set $\{1,\dots,n\}$---that is, one vertex is of degree $n-1$ and the others have degree $1$. Suppose that vertex $1$ is of degree $1$; and vertex $2$ is of degree $n-1$. Assume that $X(0)=(1,0,\dots,0)$ and consider the $2^{n-1}\times 2^{n-1}$ submatrix of the transition matrix. It can be seen that there are $2^{n-2}+1$ nonzero entries on the off-diagonal. In empirical settings, the storage of data like this becomes an issue. Even though $G$ is one of the sparsest graphs, we may not be able to calculate explicitly our desired mean first passage times for sufficiently large orders of $G$ by using \eqref{Formula:MFPT}.

\begin{example}
Consider the graph $G$ in Figure \ref{Simple_Graph_1}. Let $V=\{v_1,v_2,v_3,v_4\}$, $A=\{v_1\}$, and $B=\{v_1,v_2\}$. For simplicity, we shall remove parentheses and commas of all states in $S$. Then, $s_A=1000,\; S_A = \{1000, 1100, 1010, 1001, 1110, 1101, 1011, 1111\},$ and $S_B=\{1100,1110,1101,1111\}$. One can verify that we obtain the matrix
\begin{align*}
	T_{(\overline{S}_A\cup S_B)}=\begin{bmatrix}
	0.729 & 0.081 & 0.081 & 0.009\\
	0 & 0.729 & 0 & 0.081 \\
	0 & 0 & 0.81 & 0.09\\
	0 & 0 & 0 & 0.81
	\end{bmatrix}.
\end{align*}
This substochastic matrix represents transitions between all possible states in which $v_1$ is infected, but $v_2$ is not yet infected. Rows and columns correspond to the states $1000,1010,1001,1011$, in that order. Note that these states correspond to the set $S_A\backslash S_B$.
From \eqref{MIT}, we have $\mu_{A, B} = 7.7562$ (rounded to four decimal places), which is the expected time for susceptible individual $2$ to get infected, provided the only infected node is $1$. So, $\mu_{A, B}$ corresponds to the $(1,2)$ entry of $M_{inf}$ in \eqref{matrix:M1}. In this manner, one can calculate the matrix $M_{inf}$ (rounded up to $4$ decimal places) as follows: 

\begin{align}\label{actual MIT}
M_{inf}=\begin{bmatrix}
0 & 7.7562 & 7.7562 & 10\\
7.7562 & 0 & 7.7562 & 17.7562\\
7.7562 & 7.7562 & 0 &  17.7562\\
10 & 17.7562 & 17.7562 & 0
\end{bmatrix}.
\end{align}

\end{example}

\subsection{A sampling strategy for the $(0,1)$ model}
In this subsection, we give an alternative construction of the $(0,1)$ model, which allows us to use a sampling method for more efficient approximation of the mean infection times, and can be used to prove a nontrivial symmetry property. This construction is a discrete-time analogue of the one appearing in~\cite[Theorem~II.2]{goeringetal}, which the authors describe as `folklore in some circles'.

Let us consider a single edge $\{v_i, v_j\} \in E$ in the network. Assuming that one of its two endpoints is infected and the other is not, the time $\tau_{i,j}$ taken for the infection to spread across the edge follows a geometric distribution, i.e., $\mathbb{P}[\tau_{i,j} = k] = (1-\beta)^{k-1}\beta$ for each $k=1,2,\dots$. We call $\tau_{i,j}$ the \emph{potential infection time} over the edge $i,j$, because we are working under the assumption that the infection spreads over this edge of the network from one of its endpoints to the other.

In this new alternative construction, we first assign independently a value $\tau_{i,j}$ to each edge $\{v_i, v_j\} \in E$, sampling according to the geometric distribution described above. Once these values $\tau_{i,j}$ have been chosen, we can reconstruct the dynamics of the infection starting from an initial infected set $\hat{\mathbf{X}}(0) = s_A$ for some $A \subseteq V$: an infection happens across edge $\{v_i, v_j\}$ after time $\tau_{i,j}$ if one of the two endpoints is infected and the other is not. Hence the infection spreads from a vertex $v_k$ to another vertex $v_{\ell}$ in time $t$ equal to the graph distance $\operatorname{dist}(v_k,v_\ell)$ on the graph $G$, with weights (edge lengths) given by $\tau_{i,j}$.

The Markov process that describes the infected individuals at time $t$ in this construction is thus $\hat{\mathbf{X}}(t) = s_{\{v \in V \colon \operatorname{dist}(A,v) \leq t \}}$, where $\operatorname{dist}(A,v) = \min_{w\in A} \operatorname{dist}(w,v)$. An example of possible sample values of the potential infection times $\tau_{i,j}$ and the resulting dynamic of the process $\hat{\mathbf{X}}(t)$ is depicted in Figure~\ref{fig:construction-example}.

\begin{figure}
\begin{center}
	\begin{tikzpicture}[scale=1, baseline]
	\tikzset{enclosed/.style={draw, circle, inner sep=0pt, minimum size=.10cm, fill=black}}
	
	\node[enclosed, label={above, yshift=0cm: $1$}] (v_1) at (-2,0.5) {};
	\node[enclosed, label={below, yshift=0cm: $2$}] (v_2) at (-2,-0.5) {};
	\node[enclosed, label={above, yshift=0cm: $3$}] (v_3) at (-1,0.5) {};
	\node[enclosed, label={below, yshift=0cm: $4$}] (v_4) at (-1,-0.5) {};
	
	\draw (v_1) -- (v_2) node[midway,left,blue] {$3$};
	\draw (v_2) -- (v_3) node[pos=0.1,right,blue] {$7$};
	\draw (v_3) -- (v_1) node[midway,above,blue] {$3$};
	\draw (v_1) -- (v_4) node[pos=0.7,right,blue] {$1$};
	\end{tikzpicture}
	\quad \quad
	\begin{tabular}{ccc}
	\toprule
	\textbf{Time} $t$ & \textbf{Infected set} & $\hat{\mathbf{X}}(t)$\\
	\midrule
	0 & \{2\} & (0,1,0,0)\\
	1 & \{2\} & (0,1,0,0)\\
	2 & \{2\} & (0,1,0,0)\\
	3 & \{1,2\} & (1,1,0,0)\\
	4 & \{1,2,4\} & (1,1,0,1)\\
	5 & \{1,2,4\} & (1,1,0,1)\\
	6 & \{1,2,3,4\} & (1,1,1,1)\\
	7 & \{1,2,3,4\} & (1,1,1,1)\\
	8 & \{1,2,3,4\} & (1,1,1,1)\\
	\bottomrule
	\end{tabular}
\end{center}
\caption{A random choice of the potential infection times $\tau_{i,j}$ in our example graph, displayed in blue, and the resulting infection dynamic starting from $\hat{\mathbf{X}}(0) = s_2$. Node 1 gets infected at time 3, node 4 gets infected at time $3+1=4$, and node 3 gets infected at time $3+3=6$ via node 1. The infection does not spread across edge $\{2,3\}$,  since both endpoints are already infected at time $7$ when the potential infection is due to happen.} \label{fig:construction-example}
\end{figure}

The following proposition shows that this alternative construction provides the same result. This result could be considered intuitive by some, but we provide a formal proof.
\begin{theorem} \label{thm:equidistribution}
The stochastic processes $\mathbf{X}(t)$ and $\hat{\mathbf{X}}(t)$ (with the same initial state $\mathbf{X}(0) = \hat{\mathbf{X}}(0) = s_A$) are equidistributed.
\end{theorem}
\begin{proof}
We shall show that the transition probabilities from $\hat{\mathbf{X}}(t)$ to $\hat{\mathbf{X}}(t+1)$ coincide with those in~\eqref{transprob} and are independent of previous history. Let us condition on the state $\hat{\mathbf{X}}(t) = \mathbf{x}$, and consider the probability that $\hat{\mathbf{X}}(t+1) = \mathbf{y}$ has its $j$\textsuperscript{th} component equal to $0$, i.e., that $\operatorname{dist}(A, v_j) > t+1$. If $x_j=1$, then $\operatorname{dist}(A, v_j) \leq t$, hence it is impossible (probability $0$) that this distance is larger than $t+1$. Otherwise, let us call $n_j$ the number of edges that join $v_j$ with another node $v_i$ with $x_i=1$. For $v_j$ to stay un-infected, it must be the case that $\operatorname{dist}(A,v_i) + \tau_{i,j} > t+1$ for each such edge; and since we are assuming that $x_j=0$ it must already be the case that $\operatorname{dist}(A,v_i) + \tau_{i,j} > t$. By the memoryless property of the geometric distribution,
\[
\mathbb{P}[\tau_{i,j} > t-\operatorname{dist}(A,v_i)+1 \mid \tau_{i,j} > t-\operatorname{dist}(A,v_i)] = 1-\beta,
\]
and $\tau_{i,j}$ is independent from the potential infection times of all other edges by construction. Hence
\[
\mathbb{P}[\hat{X}_j(t+1)=0 \mid \hat{X}_i(t) = 0] = (1-\beta)^{n_j}.
\]
Hence we have proved the last two cases in~\eqref{transprob2}; the first two cases follow by difference. The product formula~\eqref{transprob} follows from the fact that each factor depends on different potential infection times $\tau_{i,j}$, which are independent.
\end{proof}

The alternate construction leads immediately to an algorithm for more efficient Monte Carlo simulation of the system, which we describe in Algorithm~\ref{algo:efficientsim}.
\begin{algorithm}
\KwData{A graph $G = (V, E)$; infection probability $\beta$; a number of samples $N$.}
\KwResult{Estimates of the MIT matrix $M_{inf}\in\mathbb{R}^{n\times n}$, and of the MIT vector $\mathbf{m} = [m_1,m_2,\dots,m_n]\in\mathbb{R}^n$.}
$M_{inf} \gets 0$; $\mathbf{m} \gets \mathbf{0}$\;
\For(\tcp*{Possibly in parallel}){$k = 1,2,\dots,N$}{
	\For{ $\{v_i,v_j\}\in E$}{
		Generate a random sample $\tau_{i,j}^{(k)}$ from a geometric distribution with parameter $\beta$\;
	}
	Compute the all-pairs shortest-path matrix $M^{(k)}$ of the graph $(V,E)$ with weights $\tau_{i,j}^{(k)}$\;
	Compute row-by-row maxima $m^{(k)}_i = \max_j M^{(k)}_{i,j}$\;
	$M_{inf} \gets M_{inf} + M^{(k)}$; $\mathbf{m} \gets \mathbf{m} + \mathbf{m}^{(k)}$\;
}
$M_{inf} \gets M_{inf} / N$; $\mathbf{m} \gets \mathbf{m} / N$\;
\caption{Efficient Monte Carlo estimation of the mean infection time matrix $M_{inf}$ (with $M_{i,j}= \mu(\{v_i\},\{v_i,v_j\})$) and vector $\mathbf{m}$ with $m_i = \mu(\{v_i\}, V)$. } \label{algo:efficientsim}
\end{algorithm}
Some remarks are in order.
\begin{itemize}
	\item The most expensive part of Algorithm~\ref{algo:efficientsim} is the all-pairs shortest-path matrix computation, for a total cost of $O(N|V||E|)$ using the algorithm in~\cite{Thorup1999UndirectedSS} where $N$ is the number of simulations. We report in Section \ref{sec:acc} on how the computation time compares with both the exact computation of \eqref{MIT} and the Monte Carlo simulations of Algorithm \ref{algo:MCS}, and also discuss accuracy and precision.
	\item The construction and the algorithm can be generalized easily to directed graphs, and to deal with edge-dependent infection probabilities $\beta_{i,j}$ (as long as they are independent from one another).
	\item From this alternate description, one can prove easily the following result, which is a discrete-time analogue of part of~\cite[Theorem~2.3]{goeringetal}.
	\begin{corollary}
 	On an undirected graph (possibly with edge-dependent infection probabilities $\beta_{i,j}$), the MIT matrix $M_{inf}$ is symmetric.
	\end{corollary}
	\begin{proof}
		The MIT matrix $M_{inf}$ is the mean over all possible samples $\tau_{i,j}^{(k)}$ of the sample MIT matrix $M^{(k)}$, which is symmetric for each $k$.
	\end{proof}
\end{itemize}


\subsection{Accuracy and computational burden of the sampling strategy}\label{sec:acc}

For actual computation of mean infection times for a network, we have to deal with transition matrix whose order is exponential with respect to the number of vertices of the network, as seen in~\eqref{MIT}. The computation is only feasible up to graphs of order $n\approx 15$, due to storage problems. So, it is necessary to turn to algorithms to estimate MITs. Hence, we compare the two Monte Carlo estimations for mean infection times, in terms of speed and accuracy.

In order to see which algorithm is more accurate, we need to have the actual MITs, so we shall consider various graphs on at most $12$ vertices. For each of the two algorithms, we report the mean $\mu$ and the variance $\sigma^2$ of relative errors for entries between the actual and estimated MIT matrices, running $300$ simulations:  
\begin{table}[h!]
\begin{center}
	\begin{tabular}{c|cc}
		\toprule
		\textbf{Graphs} & $(\mu_1,\sigma_1^2)$ & $(\mu_2,\sigma_2^2)$\\ 
		\midrule
		\Cref{Simple_Graph_1} & $(0.0376, 0.0008)$ & $(0.0333,0.0005)$ \\
		Paley graph of order $9$ & $(0.0264, 0.0006)$ & $(0.0214,0.0003)$\\
		Petersen graph & $(0.0241, 0.0005)$ & $(0.0278, 0.0005)$\\
		Star of order $12$  & $(0.0385, 0.0008)$ & $(0.0229, 0.0003)$\\
		$L(3,4)$ & $(0.0239, 0.0004)$ & $(0.0155, 0.0002)$\\
		\bottomrule
	\end{tabular}
	\caption{$(\mu_1,\sigma_1^2)$ corresponds to Algorithm \ref{algo:MCS} and $(\mu_2,\sigma_2^2)$ corresponds to Algorithm \ref{algo:efficientsim}.}\label{tab:table for accuracy and speed}
\end{center}
\end{table}

We now compare running times. The running time of Algorithm~\ref{algo:MCS} mainly depends on the infection probability---that is, the smaller $\beta$ is, the longer each simulation of the time taken for an individual $j$ to become infected takes. In contrast, the most expensive part of Algorithm~\ref{algo:efficientsim} is to compute the all-pairs shortest-path matrix, which does not depend on infection probability. For demonstration, we pick four contact networks  \cite{nr-aaai15}, which are available at \url{https://networkrepository.com}. When a network is disconnected, we choose its largest component. We measure the computation times of the MIT matrix, running $100$ simulations/samplings with different infection probabilities $0.2$, $0.1$, and $0.05$. 

\begin{table}[h!]
	\begin{center}
		\begin{tabular}{c|cc|c|c|c}
			\toprule
			\textbf{Networks} & $|V|$ & $|E|$ & $\left(t_1^{(0.2)},t_2^{(0.2)}\right)$ & $\left(t_1^{(0.1)},t_2^{(0.1)}\right)$ & $\left(t_1^{(0.05)},t_2^{(0.05)}\right)$ \\
			\midrule
			$G_1$ & $92$ & $755$ & $(8.56,0.0677)$ & $(15.30,0.0838)$ & $(29.93,0.0844)$ \\
			$G_2$ & $113$ & $2196$ & $(25.87,0.1076)$ & $(48.97,0.1267)$ & $(95.02,0.1111)$ \\
			$G_3$ & $274$ & $2124$ & $(433.8,0.2370)$ & $(911.2,0.3281)$  & $(1849.5,0.3014)$\\
			$G_4$ & $327$ & $5818$ & $(191,0.5188)$ & $(344.8,0.6023)$ & $(647,0.6193)$\\
			\bottomrule
		\end{tabular}
		\caption{The graphs $G_1,G_2,G_3$ and $G_4$ correspond to `ia-workplace-contacts', `ia-contacts\_hypertext2009', `ia-contact' and `contacts-prox-high-school-2013', respectively. We use $\left(t_1^{(p)},t_2^{(p)}\right)$ to indicate a pair of times in seconds for computing MIT matrix by $100$ simulations/samplings with Algorithms \ref{algo:MCS} and \ref{algo:efficientsim} respectively, where $p$ is the infection probability.}\label{tab:table for time}
	\end{center}
\end{table}

Table \ref{tab:table for time} shows that the approach based on the sampling strategy in Algorithm \ref{algo:efficientsim} is very efficient and scales very well with the size of the population of individuals, thus paving the way for a Markovian analysis of epidemic networks that does not suffer from the curse of dimensionality.

\section{Using mean infection time as a centrality measure} \label{sec:ranking}

Ranking nodes in a graph has many applications in epidemic models, in terms of control strategies such as testing, vaccinating, and so on (see  \cite{dudkina2021node, salathe2010dynamics}). Based on mean infection times, we shall define an indicator for ranking nodes in a graph that retains the information of the dynamics of infectious disease, and we provide several examples for comparison with other indicators. We define this measure of the `centrality' of a node using the mean infection times defined in the previous section, focusing in particular on the time it takes for the entire network to become infected---that is, the mean first passage time of the Markov chain to the absorbing state $s_{V}=(1, \ldots, 1)$.

Given a graph $G$ on $n$ nodes, by \eqref{MIT}, we define the mean infection time $m_i$ from $A=\{v_i\}$ to the set $V$ as 
\begin{align}\label{Formula:MFPT}
m_i=\mu_{A, V} = \be_{s_{v_i}}^\top\left(I-T_{(\overline{S}_{v_i}\cup s_{V})}\right)^{-1}\mathbbm{1}.
\end{align}
Then, $m_i$ indicates the expected time for all individuals to get infected, starting from an initially infected node $v_i$. We note that the author of \cite{dimitriou2006infection} investigated the mean infection time for the case that the infection probability is $1$. Moreover, as an analogous notion of the mean infection times, one may consider the \textit{cover time} of a Markov chain \cite{levin2017markov}, which is the expected value of the first time at which all the states have been visited.

\subsection{Comparison with other indicators on $2$-community networks}\label{subsec5.2:RankingNodes}

Here we compare what nodes are regarded as `most important' according to different ranking indicators. In many social networks, node degree is an important metric to determine popular nodes, i.e., individuals who typically meet a large number of other individuals during a day, or who have many social ties on online social platforms. At the same time, it is also known that in networks with strong community structures, immunization interventions targeted at individuals bridging communities (e.g., families, school classes, working environment) are more effective than those simply targeting highly-connected individuals \cite{salathe2010dynamics}. The reason for this is that regardless of the number of contacts one has, an individual bridging communities may pass the virus from one community to another community, and cause the infection to spread to a fully-susceptible community. 

In what follows, we shall consider four random-walk based indicators or centrality measures, namely, the mean infection times introduced in this section, Kemeny's constant, random walk betweenness (RWB), and random walk centrality (RWC). Here we provide formulae for RWB and RWC; see \eqref{kemeny_formula} for the description of Kemeny's constant. 

The random walk betweenness of node $i$ in a network, denoted $b_i$, measures the `expected net number of times a random walk passes through vertex $i$ on its way from a source vertex $s$ to a target vertex $t$, averaged over all $s$ and $t$', which is given \cite{newman2005measure} by
\begin{align*}
	b_i=\frac{2}{|V|(|V|-1)}\sum_{s<t}I_{i}^{(st)},
\end{align*}
where $I_{i}^{(st)}$ is the current flow through vertex $i$ starting from vertex $s$ and ending at vertex $t$. This method for computing random walk betweenness depends on considering the graph as an electrical network in which each edge has a resistance of 1, and calculating the flow through each vertex (see \cite{newman2005measure} for details). This is then shown to be equivalent to the random walk interpretation given above. 

The random walk centrality of node $i$ in a network quantifies how easily or how quickly a random walker arrives at $i$ from elsewhere in the network \cite{noh2004random}, which is given \cite{kirkland2016random} by the reciprocal of the so-called \textit{accessibility index} $\alpha_i$ of node $i$:
\begin{align*}
\alpha_i = \sum_{k\neq i} {\pi_k m_{k,i}}
\end{align*}
where $\pi_k$ is the stationary distribution for node $k$ and $m_{k,j}$ is the mean first passage time from $k$ to $i$. The accessibility index can be interpreted as the expected time for the random walker to arrive at node $i$, starting from a randomly-chosen node $j$, where this starting node is chosen with respect to the stationary distribution vector (i.e. with respect to the degree of the node). Since the random walk centrality of a node $i$ is $\frac{1}{\alpha_i}$, a large value for RWC indicates that the vertex in question is highly central, where `centrality' is understood in terms of random walks terminating at a target vertex.

Regarding Kemeny's constant, we score important nodes according to how their removal from the network affects the value of Kemeny's constant: $c_i=\mathcal{K}(G')-\mathcal{K}(G)$ where $G'$ is the graph obtained from $G$ by removing node $i$. That is, the higher the increment in Kemeny's constant after deletion of a node is, the more critical the node is (recall that large values of Kemeny's constant correspond to graphs which are not well-connected).

In the comparison that follows, we estimate MITs by Algorithm~\ref{algo:efficientsim} with $300$ samplings. For RWB, RWC and Kemeny's constant, larger values of the indicator correspond to more important nodes, while the reverse occurs for MIT (i.e., lower values of MIT correspond to more important nodes). 

As a first example, consider the lattice graph $L(6,6)$. Figure \ref{fig:onecommu} shows that all indicators agree that nodes placed in the center of the structure are the most important nodes. 

\begin{figure}[h!]
	\begin{center}
		\includegraphics[width=1\textwidth]{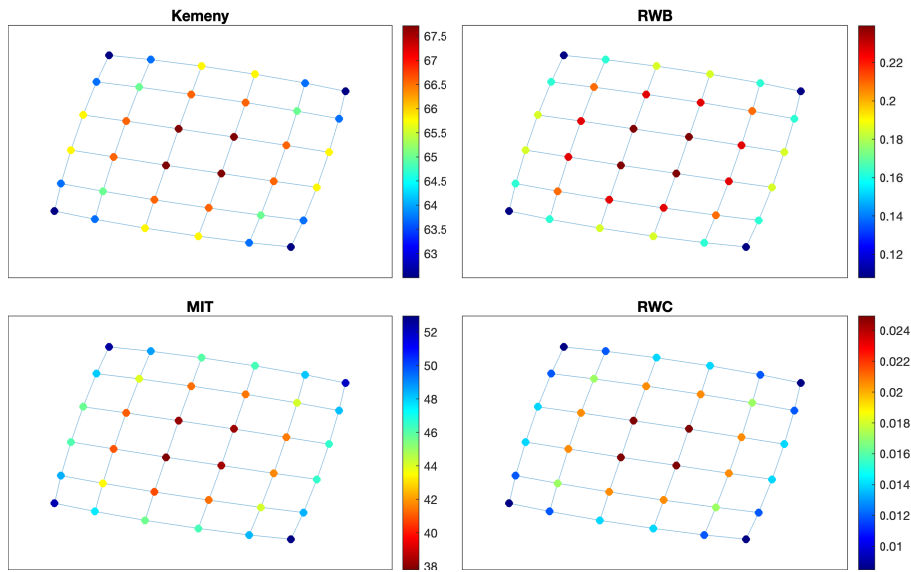}
		\caption{Ranking nodes in $L(6,6)$.}\label{fig:onecommu}
	\end{center}
\end{figure}

We now consider networks formed by two communities as follows: given two lattice graphs $L(n_1,n_1)$ and $L(n_2,n_2)$, we let an edge $e_1$ join a vertex $v_1$ at a corner of $L(n_1,n_1)$ and a vertex $v_2$ at a corner of $L(n_2,n_2)$, and we also let $e_2$ join a vertex adjacent to $v_1$ in $L(n_1,n_1)$ and a vertex adjacent to $v_2$ in $L(n_2,n_2)$. We denote by $G(n_1,n_2)$ the resulting $2$-community structure, and we shall refer to $e_1$ and $e_2$ as the community bridges. 
In particular, we now consider two types of $2$-community structures: \begin{enumerate*}[label=(\arabic*)]
	\item two groups have the same size (i.e., $G(7,7)$), and
	\item one of two groups has a larger size than the other (i.e., $G(10,3)$).
\end{enumerate*}
Let us first consider the case of $G(7,7)$. As seen in Figure \ref{fig:BigComm}, all indicators show that the most critical nodes are positioned near the community bridges, and in particular, vertices $v_1$ and $v_2$ are the most critical ones. However, RWC indicates that the third and fourth ranked nodes are not incident to $e_2$, while the other indicators do. Moreover, Kemeny's constant suggests that nodes farther away from the community bridges have similar criticality scores of the nodes near the bridges; while MITs provide more diversified scores that change gradually (in terms of criticality) as nodes are farther away from the bridges.

\begin{figure}[h!]
	\begin{center}
		\includegraphics[width=1\textwidth]{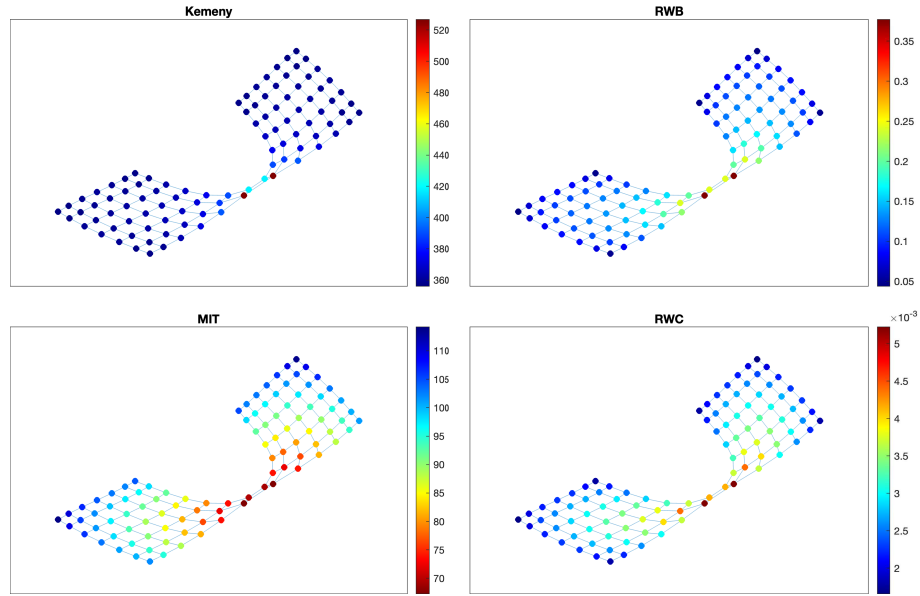}
		\caption{Ranking nodes in $G(7,7)$ based on the random walk-based indicators.}\label{fig:BigComm}
	\end{center}
\end{figure}

We finally consider the case of $G(10,3)$, shown in Figure \ref{fig:TwoComm}. The four cases now provide different indications about what should be classified as the most critical nodes. In particular, the positions of the most critical nodes according to Kemeny's constant, RWB, MIT, and RWC, respectively, vary in distance from the community bridges toward the central area of the larger community. This result appears to suggest that a criticality measure based on Kemeny's constant via node removal could be useful for detecting community bridges regardless of the sizes of the groups in a community structure (which is in accordance with \cite{yilmaz2020kemeny}). Conversely, RWC appears to rank highly the nodes that lie in the central area of the relatively larger community. MIT and RWB provide intermediate indications that take into account both the sizes of the groups and the presence of bridges (where RWB gives more importance to the community bridges than MIT). Finally, the values of the scores from MIT are spread more evenly across their (linear) scale. 

\begin{figure}[h!]
	\begin{center}
		\includegraphics[width=1\textwidth]{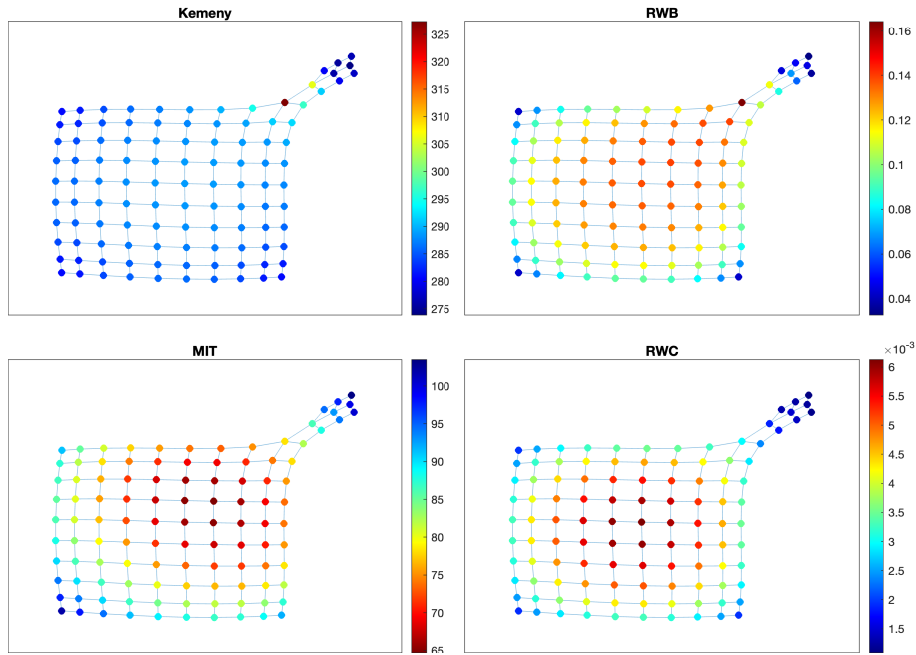}
		\caption{Ranking nodes in $G(10,3)$ based on the random walk-based indicators.}\label{fig:TwoComm}
	\end{center}
\end{figure}

	
\end{document}